\documentclass[journal]{IEEEtran}
\usepackage{cite}

\usepackage{graphicx}
\usepackage{multicol}
\usepackage{epstopdf}
\usepackage{color}

\usepackage{algorithm}
\usepackage{algorithmic}
\usepackage{amsthm,amsmath,amsfonts}
\usepackage{bbm}
\usepackage{mathrsfs}
\newtheorem{lemma}{\textbf{Lemma}}
\usepackage{bm}
\usepackage{subfigure}

\begin{document}
%
\title{Pilot-Aided Joint Time Synchronization and Channel Estimation for OTFS}
%
%
%
\author{Jiazheng Sun, Peng Yang,~\IEEEmembership{Member,~IEEE}, 
        Xianbin Cao,~\IEEEmembership{Senior Member,~IEEE},\\ 
        Zehui Xiong,~\IEEEmembership{Senior Member,~IEEE},
        Haijun Zhang,~\IEEEmembership{Fellow,~IEEE},
        and Tony Q. S. Quek,~\IEEEmembership{Fellow,~IEEE}
\thanks{

J. Sun, P. Yang, and X. Cao are with School of Electronic and Information Engineering, Beihang University, Beijing 100191, China.


Z. Xiong and T. Quek are with the Information Systems Technology and Design Pillar, Singapore University of Technology and Design, Singapore 487372. 

H. Zhang is with the Beijing Engineering and Technology Research Center for Convergence Networks and Ubiquitous Services, University of Science and Technology Beijing, Beijing 100083, China. 
}
}

\maketitle

\begin{abstract}
This letter proposes a pilot-aided joint time synchronization and channel estimation (JTSCE) algorithm for orthogonal time frequency space (OTFS) systems. Unlike existing algorithms, JTSCE employs a maximum length sequence (MLS) rather than an isolated signal as the pilot. Distinctively, JTSCE explores MLS’s autocorrelation properties to estimate timing offset and channel delay taps. After obtaining delay taps, closed-form expressions of Doppler and channel gain for each propagation path are derived. Simulation results indicate that, compared to its counterpart, JTSCE achieves better bit error rate performance, close to that with perfect time synchronization and channel state information.
\end{abstract}


%
\IEEEpeerreviewmaketitle

\section{Introduction}
%
%
%
%


\IEEEPARstart{W}{ith} {the rapid development of next-generation wireless networks, the application scenarios of wireless communications are characterized by high carrier frequencies and high mobility \cite{DBLP:journals/wcl/AbidTK24}. In such high Doppler shift scenarios, the rapidly time-varying multipath channels, i.e., doubly-selective channels, degrade the performance of widely adopted modulation schemes like orthogonal frequency division multiplexing (OFDM). This is because this type of modulation schemes is based on the assumption of time-invariant channels. To address this issue, researchers have proposed a promising orthogonal time frequency space (OTFS) modulation scheme based on the assumption of doubly-selective channels \cite{Singh2024AutoencoderBE}. OTFS characterizes channels and modulates information symbols in the two-dimensional delay-Doppler (DD) domain. Therefore, this new modulation scheme enables accurate estimation of both multipath delay and Doppler shift through channel estimation, demonstrating strong robustness in high mobility communication scenarios.}



In any practical communication system including OTFS systems, synchronization and channel estimation are crucial techniques. This paper investigates the issues of signal synchronization and channel estimation in OTFS systems.

There has been extensive research on OTFS channel estimation \cite{raviteja2019embedded}\cite{yuan2021data}\cite{hashimoto2021channel}\cite{wei2022off}, with a representative work being the embedded pilot-aided algorithm (EPA) \cite{raviteja2019embedded}. This method performed channel estimation in the DD domain and addresses the fractional Doppler issue by arranging additional guard symbols. Another significant work \cite{yuan2021data} superimposed pilot and data symbols. {It iteratively estimated channels and detected data symbols}. 
However, these methods {do not investigate the} 
synchronization problem. 
{Regarding synchronization}, the research on OTFS system synchronization is under-studied currently \cite{bayat2022time}\cite{li2023downlink}. The embedded pilot signal in \cite{raviteja2019embedded} exhibits constant amplitude with equal phase difference in the delay-time (DT) domain. Utilizing this characteristic, the authors in \cite{bayat2022time} achieved time synchronization in the DT domain. {Particularly, they utilized} an absolute timing metrics for TO estimation, which could identify the strongest path. Nevertheless, when the first path is not the strongest, this method fails to locate the first multipath component, resulting in TO estimation errors. Besides, due to the poor autocorrelation properties of the constant amplitude sequence in the DT domain, this method requires a two-step process for TO estimation in both delay and time dimensions.

To address the aforementioned issues, this letter {utilizes maximum length sequence (MLS) as the pilot signal and} proposes a joint time synchronization and channel estimation (JTSCE) algorithm for OTFS systems. 
JTSCE provides two primary advantages. In terms of synchronization, it eliminates the need for time-dimension TO estimation and improves TO estimation accuracy under specific channel conditions. In terms of channel estimation, particularly for Doppler estimation, it employs an off-grid estimation scheme, which enhances estimation accuracy in fractional Doppler cases. 

The main contributions of this letter are summarized as follows: 1) We utilize MLS as the pilot signal and correlate the received signal with the local MLS in the DT domain to achieve time synchronization without the need for time-dimension TO estimation. 2) Using the same pilot signal, we accomplish channel estimation and derive closed-form solutions for both channel gain and off-grid Doppler estimation based on 
{the correlation results obtained from time synchronization.}
3) 
We conduct extensive simulations to evaluate the accuracy of TO and channel estimation of the proposed algorithm. Simulation results verify that the proposed algorithm achieves high synchronization accuracy in doubly-selective channels. Further, in terms of channel estimation, compared to benchmark algorithms with the same embedded pilot signal power level, the proposed algorithm achieves better bit error rate (BER) performance, close to that with perfect time synchronization and channel state information (CSI).

\section{OTFS principles}
\label{sec:principles}

By performing an inverse symplectic finite Fourier transform (ISFFT) on the DD domain signal followed by a Heisenberg transform, the OTFS signal can be obtained. 
{Further}, when both the transmission pulse and the reception pulse are rectangular windows, the aforementioned signal generation process can be simplified. 
Specifically, an inverse discrete Fourier transform (IDFT) can be firstly applied to the DD domain signal along the Doppler dimension (each row of the DD grid), resulting in the DT domain signal. Secondly, the DT domain signal is serialized column-wise to obtain the OTFS signal. 

Mathematically, let ${X_{{\rm{DD}}}}\left[ {l,k} \right]$ be a DD domain signal composed of quadrature amplitude modulation (QAM) data symbols, with $N$ delay and $M$ Doppler bins, where $l=0,1,\cdots,M-1$ and $k=0,1,\cdots,N-1$ are the delay and Doppler indices, respectively. The DT domain signal can be represented as ${X_{{\rm{DT}}}}\left[ {l,n} \right] = \frac{1}{{\sqrt N }}\sum\nolimits_{k' = 0}^{N - 1} {{X_{{\rm{DD}}}}\left[ {l,k'} \right]{{\rm{e}}^{{\rm{j}}\frac{{2\pi }}{N}k'n}}} $, where $n = 0,1, \cdots N - 1$ is the time index. To obtain the OTFS signal, the DT domain signal needs to be serialized column-wise, that is, $s\left[ {n'} \right] = {X_{{\rm{DT}}}}\left[ {l,n} \right]$ with $n' = l + nM$.
Before transmitting the signal, {a guard interval (GI) is added to eliminate} inter-symbol interference (ISI) between OTFS symbols. 
This letter adopts reduced cyclic prefix (RCP) as the GI. That is, a cyclic prefix (CP) of length ${L_{{\rm{RCP}}}}$ is added to the beginning of each OTFS block, containing $MN$ samples. The length ${L_{{\rm{RCP}}}}$ should be greater than the delay spread of the channel.

For a doubly-selective channel, we denote its response to an impulse with delay 
$\tau$ and Doppler shift $\nu$ as $h\left( {\tau ,\nu } \right)$. The received signal $r\left( t \right)$
can then be expressed as 
\begin{equation}\label{eq:continuous_io_time}
    r\left(t\right)=\iint h\left(\tau,\nu\right)s\left(t-\tau\right)\mathrm{e}^{\mathrm{j}2\pi\nu\left(t-\tau\right)}\mathrm{d}\nu\mathrm{d}\tau.
\end{equation}

Sampling $r(t)$ with delay and Doppler resolutions of $\frac{T}{M}$ and$\frac{{\Delta f}}{N}$, respectively, i.e., $t = \frac{{n'T}}{M}$, $\tau  = \frac{{lT}}{M}$ and $\nu  = \frac{{k\Delta f}}{N}$, we can obtain a discrete form of the time domain input-output relationship, i.e.,
\begin{equation}\label{eq:discrete_io_time}
r\left[ {n'} \right] = \sum\limits_{i = 1}^P {{h_i}s\left[ {n' - {l_i}} \right]{{\rm{e}}^{{\rm{j}}\frac{{2{\rm{\pi }}}}{{MN}}{k_i}\left( {n' - {l_i}} \right)}}},    
\end{equation}
where $P$ is the number of paths of the channel, $l_i$, $k_i$, and $h_i=h\left[l_i,k_i\right]$ are the delay tap, Doppler tap, and {channel} gain of the $i$-th path, respectively.

By performing serial-to-parallel conversion on {$s\left[n'-l_i\right]$ and $r\left[n'\right]$ in (\ref{eq:discrete_io_time})}
, we can obtain the input-output relationship in the DT domain. Let $n' = nM + l$ (assuming ${l_i} < l$) and $r\left[ {n'} \right] = {Y_{{\rm{DT}}}}\left[ {l,n} \right]$, the relationship can be expressed as
\begin{equation}\label{eq:discrete_io_DT}
    {Y_{{\rm{DT}}}}\left[ {l,n} \right] = \sum\limits_{i = 1}^P {{h_i}{X_{{\rm{DT}}}}\left[ {l - {l_i},n} \right]{{\rm{e}}^{{\rm{j}}\frac{{2{\rm{\pi }}}}{{MN}}{k_i}\left( {l - {l_i}} \right)}}{{\rm{e}}^{{\rm{j}}\frac{{2{\rm{\pi }}}}{N}{k_i}n}}}.
\end{equation}

The above derivation assumes no TO. Therefore, to recover the transmitted signal $s\left[ {n'} \right]$, 
{one should first achieve time synchronization and then 
{obtain channel parameters set $\mathcal{H} = \left\{ {\left( {{l_i},{k_i},{h_i}} \right);i = 1,2, \cdots ,P} \right\}$ by channel estimation.}

\section{Proposed joint time synchronization and channel estimation algorithm}
In this section, we propose the JTSCE algorithm for OTFS systems. {OTFS systems can treat carrier frequency offset (CFO) caused by Doppler effects and local oscillator mismatch as a part of channel response \cite{hong2022delay}. Therefore, during the synchronization phase, we only consider TO estimation. The task of frequency estimation can be handled during channel estimation.}
{Next, we will estimate TO, delay, Doppler, and channel gain in two steps to achieve JTSCE.}


\subsection{TO \& delay estimation}
We first estimate the TO and delay.
{The proposed algorithm deploys MLS in the DT domain, which eliminates the need for time-dimension TO estimation, and can estimate all delay taps of the channel. Additionally, it maintains high TO estimation accuracy even if the 1$^{\rm st}$ path of the channel is not strongest.}

The TO estimation method proposed in \cite{bayat2022time} requires {an additional step} for time-dimension TO estimation. This is largely because the constant amplitude linear phase structure explored in \cite{bayat2022time} has poor autocorrelation properties, making the timing metric insensitive to time-dimension TO. 

To tackle this issue, this letter leverages MLS with good autocorrelation properties in the DT domain for TO estimation. An MLS is a pseudo-random sequence of length ${2^p} - 1$, where $p$ is a positive integer. Typically, $N$ is chosen as ${2^p}$. Therefore, we generate a bipolar MLS ${x_{{\rm{MLS}}}}\left[ n \right]$ of length $N-1$ with total power of ${P_{{\rm{MLS}}}}$ (sequence values are $ \pm \sqrt {\frac{{{P_{{\rm{MLS}}}}}}{{N - 1}}} $), where $n = 0,1, \cdots ,N - 2$. A zero is appended to the end of ${x_{{\rm{MLS}}}}\left[ n \right]$ to obtain $\tilde{x}_{\mathrm{MLS}}\left[n\right]=\begin{cases}x_{\mathrm{MLS}}\left[n\right],&n=0,1,\cdots,N-2,\\0,&n=N-1\end{cases}$. 

Next, a $N$-point discrete Fourier transform (DFT) is performed on ${\tilde x_{{\rm{MLS}}}}\left[ n \right]$ to obatain ${\tilde X_{{\rm{MLS}}}}\left[ k \right] = {\rm{DFT}}\left\{ {{{\tilde x}_{{\rm{MLS}}}}\left[ n \right]} \right\}$. ${\tilde X_{{\rm{MLS}}}}\left[ k \right]$ is then placed in the ${l_{{\rm{MLS}}}}$-th row of the DD domain grid, as shown in Fig. \ref{fig:pilot_pattern}. To prevent interference between data symbols and the MLS, guard symbols are deployed between them. The guard symbols are set to $0$, with a width $2L + 1$ where $L$ is no less than 
{the delay tap corresponding to the channel's delay spread} \cite{raviteja2019embedded}. 

\begin{figure}[!t]
    \centering
    \includegraphics[width=0.45\textwidth]{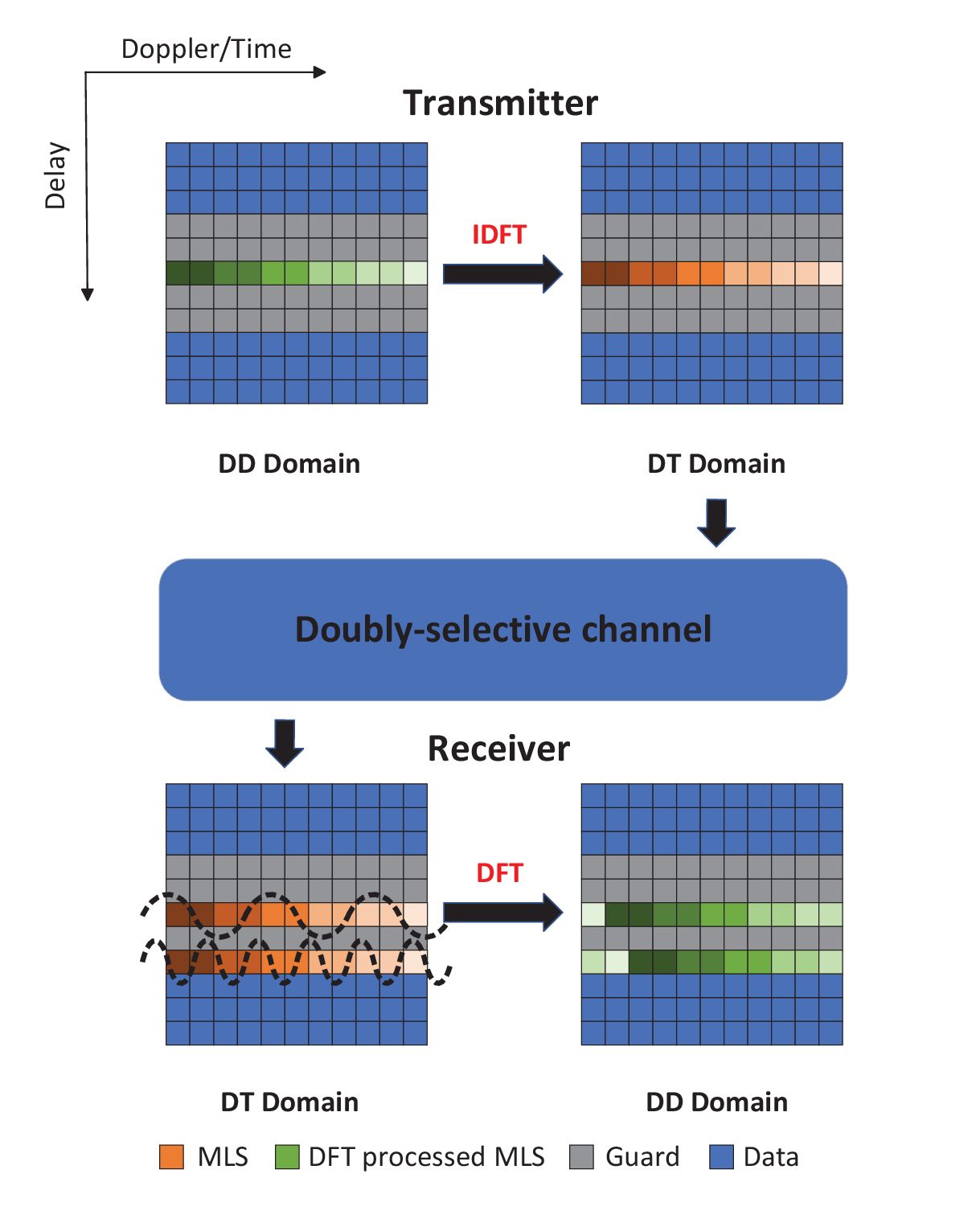}
    \caption{Pilot patterns {for both} transmitter and receiver.}
    \label{fig:pilot_pattern}
\end{figure}

Based on the input-output relationship derived in section \ref{sec:principles}, we can derive the {signal} expression for the received MLS at the receiver, as illustrated in Fig. \ref{fig:pilot_pattern}. According to (\ref{eq:discrete_io_DT}), the MLS in the received DT domain signal will be spread across different rows with delay ${l_{{\rm{MLS}}}} + {l_i}$. Further, the mathematical expression for the sequence in the $\left({l_{{\rm{MLS}}}} + {l_i}\right)$-th row of the DT grid is given by $y_{{\rm{MLS}}}^{\left( {{l_i}} \right)}\left[ n \right] = {h_i}{{\rm{e}}^{{\rm{j}}\frac{{2{\rm{\pi }}}}{{MN}}{k_i}{l_{{\rm{MLS}}}}}} {\tilde x_{{\rm{MLS}}}}\left[ n \right]{{\rm{e}}^{{\rm{j}}\frac{{2{\rm{\pi }}}}{N}{k_i}n}}$. 

It can be observed that the non-zero rows in the received DT grid are in the form of MLS multiplied by a complex exponential sequence. Therefore, by utilizing the good autocorrelation properties of MLS to locate $y_{{\rm{MLS}}}^{\left( {{l_i}} \right)}\left[ n \right]$ and obtain ${l_i}$, we can complete the estimation of TO and delay.

To locate $y_{{\rm{MLS}}}^{\left( {{l_i}} \right)}\left[ n \right]$ in the DT domain, we first extract a portion of the received signal $r\left[ {n'} \right]$ that can serve as a row in the DT domain grid. 
{We denote the row} as ${Y_{\tilde n}}\left[ n \right] = r\left[ {\tilde n + nM} \right]$, 
and use the local MLS to correlate with the corresponding part of ${Y_{\tilde n}}\left[ n \right]$. This process can be mathematically represented as ${q_{\tilde n}}\left[ n \right] = {Y_{\tilde n}}\left[ n \right] {\tilde x_{{\rm{MLS}}}}\left[ n \right]$. 

When locating $y_{{\rm{MLS}}}^{\left( {{l_i}} \right)}\left[ n \right]$, we have ${Y_{\tilde n}}\left[ n \right] = y_{{\rm{MLS}}}^{\left( {{l_i}} \right)}\left[ n \right]$ and ${q_{\tilde n}}\left[ n \right] = y_{{\rm{MLS}}}^{\left( {{l_i}} \right)}\left[ n \right] {\tilde x_{{\rm{MLS}}}}\left[ n \right] = {h_i}{{\rm{e}}^{{\rm{j}}\frac{{2{\rm{\pi }}}}{{MN}}{k_i}{l_{{\rm{MLS}}}}}} \tilde x_{{\rm{MLS}}}^2\left[ n \right]{{\rm{e}}^{{\rm{j}}\frac{{2{\rm{\pi }}}}{N}{k_i}n}}$. For the bipolar MLS ${\tilde x_{{\rm{MLS}}}}\left[ n \right]$, we have $\tilde{x}_{\mathrm{MLS}}^{2}\left[n\right]=\begin{cases}\frac{P_{\mathrm{MLS}}}{N-1},&n=0,1,\cdots,N-2,\\0,&n=N-1,\end{cases}$ and $q_{\tilde{n}}\left[n\right]=\begin{cases}\frac{P_{\mathrm{MLS}}}{N-1} h_i\mathrm{e}^{\mathrm{j}\frac{2\pi}{MN}k_il_{\mathrm{MLS}}}\mathrm{e}^{\mathrm{j}\frac{2\pi}{N}k_in},&n=0,1,\cdots,N-2,\\0,&n=N-1\end{cases}$.
Then, we perform an $N$-point DFT on ${q_{\tilde n}}\left[ n \right]$ to yield ${Q_{\tilde n}}\left[ k \right]$.
Since ${q_{\tilde n}}\left[ n \right]$ is a complex exponential sequence, the magnitude $\left| {{Q_{\tilde n}}\left[ k \right]} \right|$ will exhibit a distinct peak. When ${k_i}$ is an integer, the peak occurs at $k = {k_i}$, and its value is $\left| {{Q_{\tilde n}}\left[ k \right]} \right| = \frac{{\left| {{h_i}} \right|}}{{\sqrt N }}{P_{{\rm{MLS}}}}$. According to this principle, the proposed algorithm can detect all paths in the channel.

Next, we analyze the case {of} $y_{{\rm{MLS}}}^{\left( {{l_i}} \right)}\left[ n \right]$ not being located {accurately}, i.e., there is a TO. 
{In this case,} we decompose TO $\theta$ as $\theta  = {\theta _{\rm{d}}} + M{\theta _{\rm{t}}}$, where ${\theta _{\rm{d}}}$ and ${\theta _{\rm{t}}}$ represent delay and time dimensions TO, respectively. When there is a delay dimension TO, i.e., ${\theta _{\rm{d}}} \ne 0$, ${Y_{\tilde n}}\left[ n \right]$ does not correlate with ${\tilde x_{{\rm{MLS}}}}\left[ n \right]$. As a result, ${Q_{\tilde n}}\left[ k \right]$ will not exhibit a distinct peak. For the case {that} only the time-dimension TO exists, i.e., ${\theta _{\rm{t}}} \ne 0$ and ${\theta _{\rm{d}}} = 0$, the following lemma explains why the proposed algorithm does not require additional time-dimension TO estimation.


\begin{lemma}\label{lemma:1}
    \rm{When ${\theta _{\rm{t}}} \ne 0$ and ${\theta _{\rm{d}}} = 0$, $\left| {{Q_{\tilde n}}\left[ k \right]} \right|$ will not exhibit a distinct peak, i.e., 
    \begin{equation}
        \mathop {\max }\limits_{k = 0,1, \cdots N - 1} \left\{ {\left| {{Q_{\tilde n}}\left[ k \right]} \right|} \right\} \ll \frac{{\left| {{h_i}} \right|}}{{\sqrt N }}{P_{{\rm{MLS}}}}.
    \end{equation}
    }
\end{lemma}
\begin{proof}
    Please refer to the Appendix. 
\end{proof}
The above lemma tells that, due to the good autocorrelation properties of MLS, even if ${\theta _{\rm{t}}}$ is very small, making ${q_{\tilde n}}\left[ n \right]$ and $y_{{\rm{MLS}}}^{\left( {{l_i}} \right)}\left[ n \right]$ nearly identical except for a slight shift, $\left| {{Q_{\tilde n}}\left[ k \right]} \right|$ will not exhibit a distinct peak. Therefore, {the proposed algorithm} does not require additional time-dimension TO estimation. 

Based on the above analysis, we can estimate the delay taps and TO by performing normalized threshold detection on the maximum magnitude of DFT of ${q_{\tilde n}}\left[ n \right]$. {The detailed steps are summarized in the latter Algorithm \ref{alg:joint}.}

{We then discuss how to determine} 
the detection threshold $\mathcal{T}$. The value of $\mathcal{T}$ is determined based on the ``non-peak" condition. In the presence of TO, ${Y_{\tilde n}}\left[ n \right]$ is usually data symbols with $E\left\{ {{Y_{\tilde n}}\left[ n \right]} \right\} = 0$. Therefore, $E\left\{ {{q_{\tilde n}}\left[ n \right]} \right\} = E\left\{ {{Y_{\tilde n}}\left[ n \right]} \right\} {\tilde x_{{\rm{MLS}}}}\left[ n \right] = 0$ and ${Q_{\tilde n}}\left[ k \right] = \sum\nolimits_{k = 0}^{N - 1} {{q_{\tilde n}}\left[ n \right]{{\rm{e}}^{ - {\rm{j}}\frac{{2{\rm{\pi }}nk}}{N}}}} $ with $E\left\{ {{Q_{\tilde n}}\left[ k \right]} \right\} = 0$. According to the central limit theorem, ${Q_{\tilde n}}\left[ k \right]$ approximately follows a complex Gaussian distribution with zero mean and ${\sigma ^2}$ variance. $\left| {{Q_{\tilde n}}\left[ k \right]} \right|$ follows a Rayleigh distribution with $\sqrt {\frac{\pi }{2}} \sigma $ mean and $\frac{{4 - \pi }}{2}{\sigma ^2}$ variance. Besides, $\sum\nolimits_{k = 0}^{N - 1} {\left| {{Q_{\tilde n}}\left[ k \right]} \right|}$ follows a complex Gaussian distribution with $N\sqrt {\frac{\pi }{2}} \sigma $ mean and $N\frac{{4 - \pi }}{2}{\sigma ^2}$ variance. Considering the complex Gaussian white noise, {the above three distributions} still hold, but the value of {$\sigma$ changes}. 
Given a very small probability $P$, there exist parameters ${\beta _1},{\beta _2}$ such that $\Pr \left\{ {\alpha  > \frac{{\sqrt {\frac{\pi }{2}} \sigma  + {\beta _1}\sqrt {\frac{{4 - \pi }}{2}} \sigma }}{{N\sqrt {\frac{\pi }{2}} \sigma  - {\beta _2}N\sqrt {\frac{{4 - \pi }}{2}} \sigma }}} \right\} = P$. Thus, we can set $\mathcal{T} = \frac{1}{N} \cdot \frac{{\sqrt {\frac{\pi }{2}}  + {\beta _1}\sqrt {\frac{{4 - \pi }}{2}} }}{{\sqrt {\frac{\pi }{2}}  - {\beta _2}\sqrt {\frac{{4 - \pi }}{2}} }}$, where $\frac{{\sqrt {\frac{\pi }{2}}  + {\beta _1}\sqrt {\frac{{4 - \pi }}{2}} }}{{\sqrt {\frac{\pi }{2}}  - {\beta _2}\sqrt {\frac{{4 - \pi }}{2}} }}$ can be obtained experimentally. This value is only related to ${\beta _1},{\beta _2}$, i.e., the given probability $P$. Hence, the threshold is applicable to different signal-to-noise ratios (SNRs) and OTFS sizes.
\begin{figure}[h]
    \centering
    \subfigure[Designed TO estimation scheme]{
    \label{fig:snapshPot_subfig:joint}
    \includegraphics[width=0.4\textwidth]{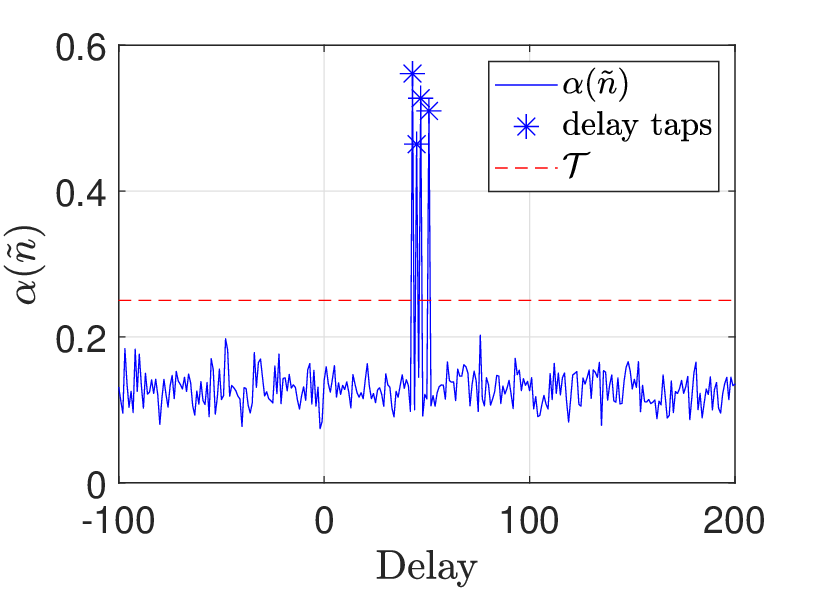}
    }
    \subfigure[TO estimation method in \cite{bayat2022time}]{
    \label{fig:snapshot_subfig:bayat}
    \includegraphics[width=0.4\textwidth]{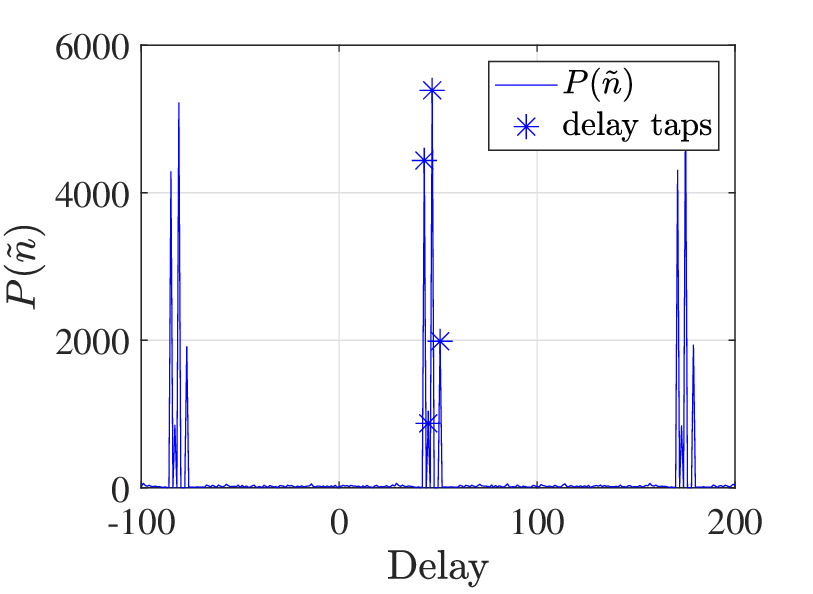}
    }
    \caption{A snapshot of $\alpha \left( {\tilde n} \right)$ and $P \left( {\tilde n} \right)$ for the designed scheme and the method in \cite{bayat2022time} at ${\rm{SN}}{{\rm{R}}_{\rm{d}}} = 10{\rm{\ dB}}$ with $M = 128$ and $N = 32$.}
    \label{fig:snapshot}
\end{figure}

Fig. \ref{fig:snapshot} depicts a snapshot of the timing metric $\alpha \left( {\tilde n} \right)$ and correlation function $P \left( {\tilde n} \right)$, used for TO estimation in the proposed algorithm and \cite{bayat2022time} respectively. The parameters for OTFS systems are $M = 128$, $N = 32$, data SNR ${\rm{SN}}{{\rm{R}}_{\rm{d}}} = 10{\rm{\ dB}}$, MLS SNR ${\rm{SN}}{{\rm{R}}_{\rm{M}}} = 30{\rm{\ dB}}$, pilot SNR \cite{bayat2022time} ${\rm{SN}}{{\rm{R}}_{\rm{p}}} = 44.9{\rm{\ dB}}$. The detection threshold $\mathcal{T}=8.0/N$ is set reasonably {and helps to effectively capture the delay taps, as seen in Fig. \ref{fig:snapshPot_subfig:joint}.}
{We can also observe from Figs. \ref{fig:snapshPot_subfig:joint} and \ref{fig:snapshot_subfig:bayat} that
the obtained $\alpha \left( {\tilde n} \right)$ by the proposed algorithm does not exhibit a peak at $\theta  = M{\theta _{\rm{t}}}$. As a result, the proposed algorithm does not require} additional TO estimation. Further, as observed in Fig. \ref{fig:snapshot_subfig:bayat}, the method in \cite{bayat2022time} suffers from TO estimation errors due to the first path of the channel not being the strongest path. In contrast, the proposed algorithm accurately estimates TO by capturing all paths of the channel.

\subsection{Doppler \& Channel Gain Estimation}
{After obtaining the TO $\theta$ and channel delay taps,  we then estimate the Doppler taps and channel gains. For any path in the channel with delay ${l_i}$, let ${\tilde n_i} = \theta  + {L_{{\rm{RCP}}}} + {l_{{\rm{MLS}}}} + {l_i}$, we have
$q_{\tilde n_i}\left[n\right]=\begin{cases}\frac{P_{\mathrm{MLS}}}{N-1} h_i\cdot\mathrm{e}^{\mathrm{j}\frac{2\pi k_in}{N}},&n=0,1,\cdots,N-2,\\0,&n=N-1\end{cases}$.
}
To estimate the Doppler tap ${k_i}$, a direct approach is to estimate it based on the location of the peak in  ${Q_{{{\tilde n}_i}}}\left[ k \right]$, which is feasible when ${k_i}$ is an integer. However, when ${k_i}$ is fractional, spectral leakage will greatly affect accurate estimation of ${k_i}$. Given that ${q_{{{\tilde n}_i}}}\left[ n \right]$ is a complex exponential sequence, this letter estimates ${k_i}$ using the following expression: ${\hat k_i} = \frac{N}{{\left( {N - 2} \right) \cdot 2{\rm{\pi }}}}\sum\limits_{n = 0}^{N - 3} {\angle \left( {{q_{{{\tilde n}_i}}}\left[ {n + 1} \right]q_{{{\tilde n}_i}}^ * \left[ n \right]} \right)} $, where $\angle$ denotes the phase angle operation. After obtaining ${\hat k_i}$, the estimated gain ${h_i}$ can be computed by the following close-formed expression: ${\hat h_i} = \frac{1}{{{P_{{\rm{MLS}}}}}}\sum\limits_{n = 0}^{N - 2} {{q_{{{\tilde n}_i}}}\left[ n \right] \cdot {{\rm{e}}^{{\rm{ - j}}\frac{{2{\rm{\pi }}{{\hat k}_i}n}}{N}}}} $. 

Finally, on the basis of the above derivations, we can summarize the main steps of time synchronization and channel estimation in Algorithm \ref{alg:joint}.
\begin{algorithm}[!htp]
    \caption{Joint time synchronization and channel estimation, JTSCE}
    \label{alg:joint}
    \begin{algorithmic}[1]
        \STATE {\textbf{Input:}} 
        {received signal $r\left[n'\right]$, local MLS ${\tilde x_{{\rm{MLS}}}}\left[ n \right]$}, 
        threshold $\mathcal{T}$, ${l_{{\rm{MLS}}}}$, delay spread $L$, and RCP length ${L_{{\rm{RCP}}}}$.
        \STATE {\textbf{Output:}} timing offset $\hat \theta $ and channel parameters set ${\hat {\mathcal H}}$.
        \STATE {\textbf{Initialize:}} ${\hat {\mathcal H}} = \emptyset $, $\tilde n = 0$, ${\rm{stop}} = \infty $, and ${\rm{trigger}} = {\rm{True}}$.
        \WHILE{$\tilde n < {\rm{stop}}$}
        \STATE Let
        {${Y_{\tilde n}}\left[ n \right] = r\left[ {\tilde n + nM} \right]$, ${q_{\tilde n}}\left[ n \right] = {Y_{\tilde n}}\left[ n \right] {\tilde x_{{\rm{MLS}}}}\left[ n \right]$, and} ${Q_{\tilde n}}\left[ k \right] = {\rm{DFT}}\left\{ {{q_{\tilde n}}\left[ n \right]} \right\}$.
        \STATE Calculate the timing metric 
        $\alpha \left( {\tilde n} \right)$, defined as  
        $\alpha \left( {\tilde n} \right) = \frac{{\mathop {\max }\limits_{k = 0,1, \cdots N - 1} \left\{ {\left| {{Q_{\tilde n}}\left[ k \right]} \right|} \right\}}}{{\sum\nolimits_{k = 0}^{N - 1} {\left| {{Q_{\tilde n}}\left[ k \right]} \right|} }}$.
        \IF{$\alpha \left( {\tilde n} \right) > \mathcal{T}$}
        \IF{${\rm{trigger}} = {\rm{True}}$}
        \STATE ${\rm{stop}} = \tilde n + L + 1$, ${\rm{trigger}} = {\rm{False}}$, and $\hat \theta  = \tilde n - {l_{{\rm{MLS}}}} - {L_{{\rm{RCP}}}}$.
        \ENDIF
        \STATE Get $\left( {\hat l,\hat k,\hat h} \right)$ through $\hat l = \tilde n - \hat \theta  - {L_{{\rm{RCP}}}} - {l_{{\rm{MLS}}}}$, $\hat k = \frac{N}{{\left( {N - 2} \right) \cdot 2\pi }}\sum\limits_{n = 0}^{N - 3} {\angle \left( {{q_{\tilde n}}\left[ {n + 1} \right]q_{\tilde n}^*\left[ n \right]} \right)}$, and $\hat h = \frac{1}{{{P_{{\rm{MLS}}}}}}\sum\limits_{n = 0}^{N - 2} {{q_{\tilde n}}\left[ n \right] \cdot {{\rm{e}}^{{\rm{ - j}}\frac{{2\pi \hat kn}}{N}}}} $. Then append $\left( {\hat l,\hat k,\hat h} \right)$ to ${\hat {\mathcal H}}$.
        \ENDIF
        \STATE $\tilde n = \tilde n + 1$.
        \ENDWHILE
    \end{algorithmic}
\end{algorithm}

\section{Simulation results}
In this section, we analyze the performance of the proposed algorithm through simulations. The default system parameters are set as follows: carrier frequency 8 GHz, subcarrier spacing 15 kHz, $M = 128$, $N = 32$, RCP length ${L_{{\rm{RCP}}}} = M/4$, detection threshold $\mathcal{T}=8.0/N$, and 4-QAM signals. We refer to \cite{raviteja2019embedded} for the definitions of average data SNR $\begin{aligned}{\rm{SN}}{{\rm{R}}_{\rm{d}}} = {\rm{E}}\left\{ {{{\left| {{x_{\rm{d}}}} \right|}^2}} \right\}{\rm{/}}{\sigma ^2}\end{aligned}$ and average MLS SNR $\begin{aligned}{\rm{SN}}{{\rm{R}}_{\rm{M}}} = {\rm{E}}\left\{ {{{\left| {{x_{{\rm{MLS}}}}} \right|}^2}} \right\}{\rm{/}}{\sigma ^2}\end{aligned}$. 
The channel model in \cite{yuan2020simple} is utilized. The maximum speed of a mobile user is set to $v = {\rm{250\ Kmph}}$, corresponding to a maximum Doppler shift index ${k_{{\nu _{\max }}}} = 4$ , and we assume a maximum delay index ${l_{{\tau _{\max }}}} = 10$.

In Fig. \ref{fig:to_delay_acc}, we investigate the performance of the designed TO and delay estimation {scheme}. Specifically, under the condition of ${\rm{SN}}{{\rm{R}}_{\rm{d}}} \in \{10,15,20\}{\rm{\ dB}}$, we conduct simulations for each ${\rm{SN}}{{\rm{R}}_{\rm{M}}}$, measuring the performance of the estimation {scheme} by the proportion of accurately estimated TO and channel delay taps. In Fig. \ref{fig:to_delay_acc}, the dashed line represents the proportion of accurately estimated TO, while the solid line represents the proportion of accurately estimated TO and delay taps. The results show that the estimation accuracy is independent of ${\rm{SN}}{{\rm{R}}_{\rm{d}}}$ and positively correlated with ${\rm{SN}}{{\rm{R}}_{\rm{M}}}$. This is because the added guard symbols ensure that the data symbols and MLS do not interfere with each other. Further, when ${\rm{SN}}{{\rm{R}}_{\rm{M}}} \geq 25{\rm{\ dB}}$, the designed TO and delay estimation {scheme} achieves accurate TO estimation. Although there are some errors in the delay estimation, it also reaches a high level of accuracy. 
The sources of error in the delay estimation are twofold: 1) Missed detection: paths with smaller channel gains have correlation peaks below the threshold and are missed. 2) False detection: spurious peaks above the threshold are falsely identified as paths. Missed detection decreases with increasing ${\rm{SN}}{{\rm{R}}_{\rm{M}}}$, while false detection is independent of ${\rm{SN}}{{\rm{R}}_{\rm{M}}}$. Therefore, the high accuracy at ${\rm{SN}}{{\rm{R}}_{\rm{M}}} = 40{\rm{\ dB}}$ indicates that missed detection is the main source of delay estimation error. Paths with smaller channel gains have a relatively minor impact on the accuracy of channel estimation. Thus, the missed detection of these paths is acceptable for the decision process to some extent. This will be verified in the subsequent BER simulations.
\begin{figure}[h]
    \centering
    \includegraphics[width=0.4\textwidth]{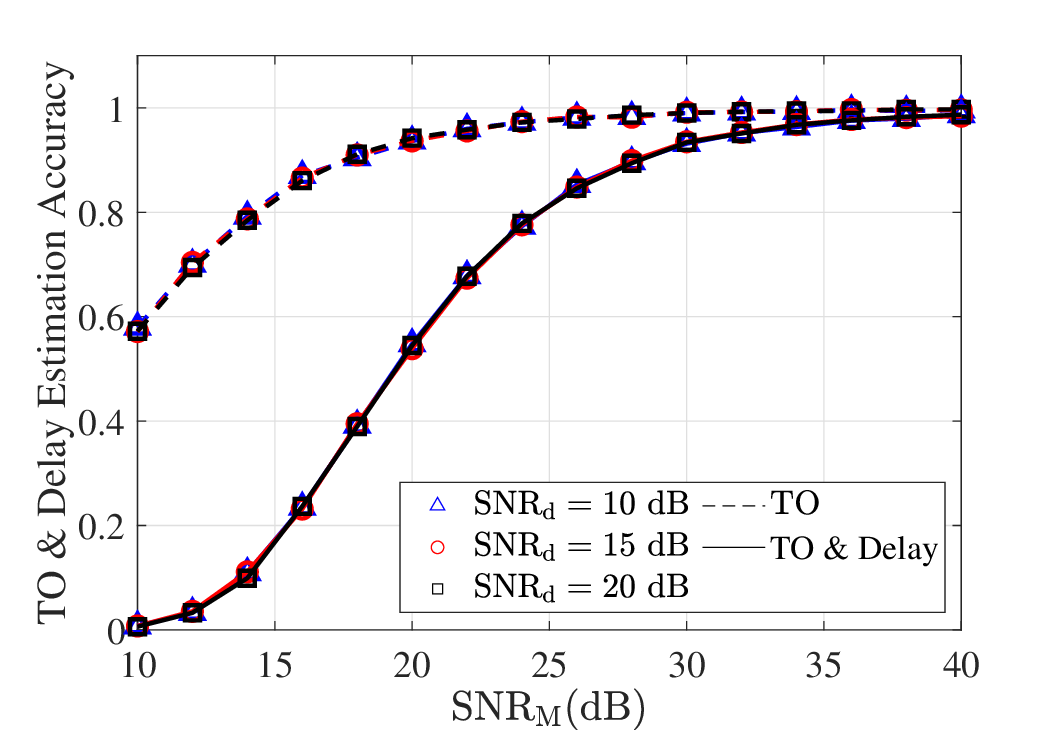}
    \caption{The accuracy of TO and delay estimation of the proposed algorithm.}
    \label{fig:to_delay_acc}
\end{figure}
 
Next, we examine the impact of threshold $\mathcal{T}$ on the estimation of TO and delay. Fig. \ref{fig:threshold} illustrates the relationship between estimation accuracy and threshold $\mathcal{T}$ under conditions of ${\rm{SN}}{{\rm{R}}_{\rm{M}}} = 25{\rm{\ dB}}$ and $35{\rm{\ dB}}$, respectively. The previous simulation results indicate that the estimation performance is independent of ${\rm{SN}}{{\rm{R}}_{\rm{d}}}$. The following simulations are 
conducted under the assumption of ${\rm{SN}}{{\rm{R}}_{\rm{d}}} = 20{\rm{\ dB}}$. 
The simulation results shown in Fig. \ref{fig:threshold} reveal that $\mathcal{T} = 8.0/N$ has a significant advantage in low MLS power scenarios and maintains good accuracy under high MLS power conditions. This is because missed detection is more likely {to happen} at low MLS power levels, where a relatively lower threshold is more beneficial. In summary, for the given system parameters, the optimal threshold approximates $\mathcal{T} = 8.0/N$.
\begin{figure}[h]
    \centering
    \subfigure[${\rm{SN}}{{\rm{R}}_{\rm{M}}} = 25{\rm{\ dB}}$]{
    \label{fig:subfig:thr25}
    \includegraphics[width=0.4\textwidth]{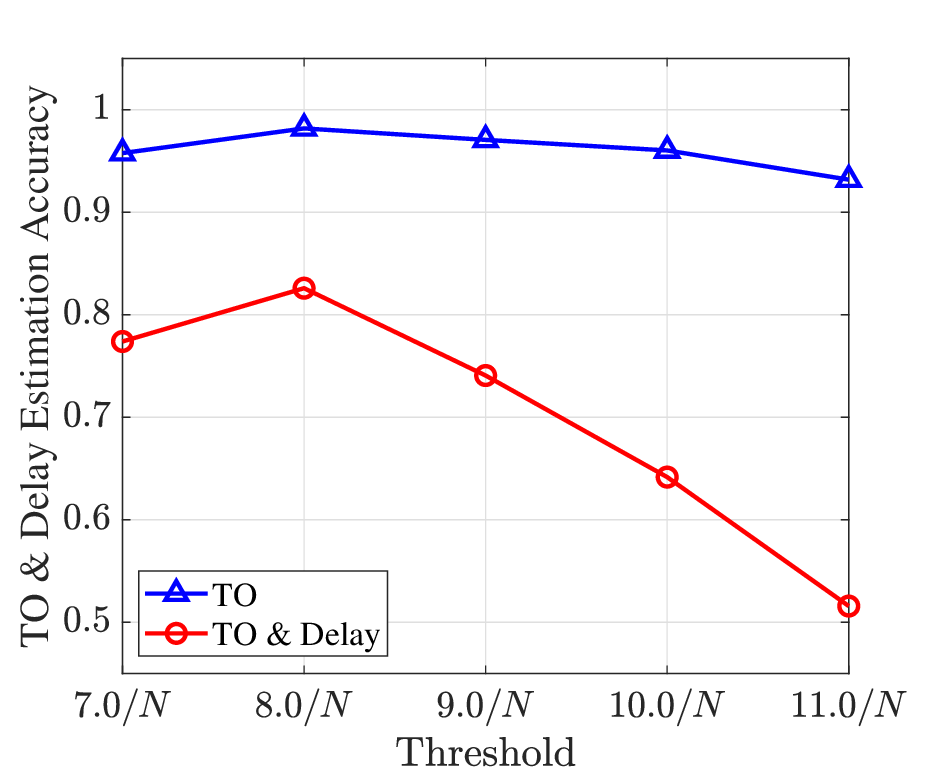}
    }
    \subfigure[${\rm{SN}}{{\rm{R}}_{\rm{M}}} = 35{\rm{\ dB}}$]{
    \label{fig:subfig:thr35}
    \includegraphics[width=0.4\textwidth]{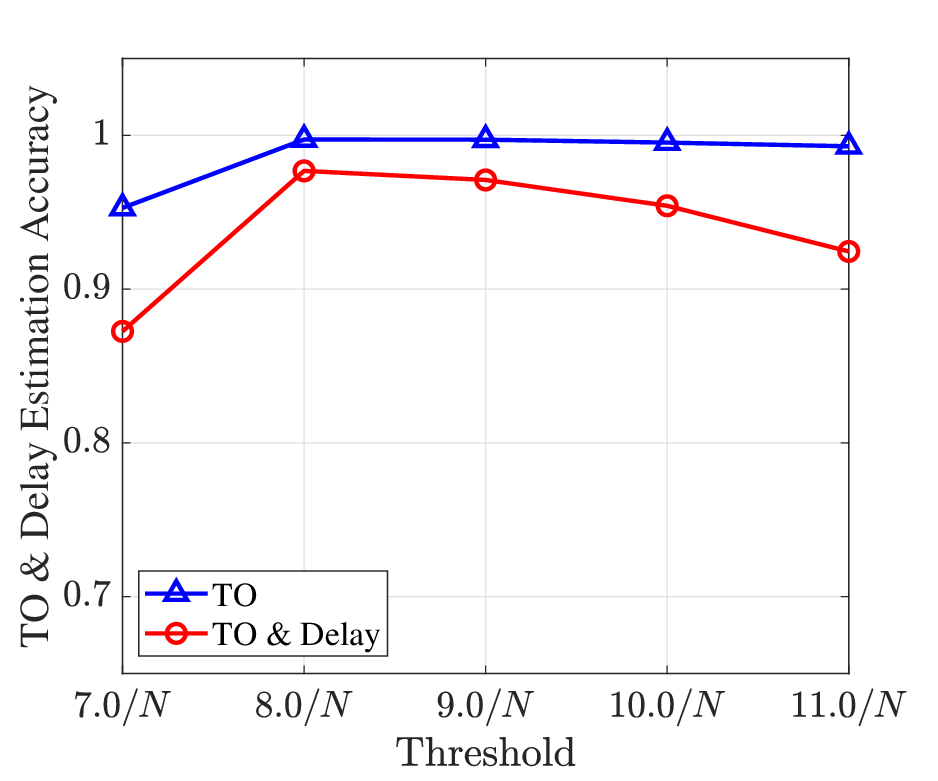}
    }
    \caption{The accuracy of TO and delay estimation versus thresholds at ${\rm{SN}}{{\rm{R}}_{\rm{M}}} = 25{\rm{\ dB}}$ and $35{\rm{\ dB}}$.}
    \label{fig:threshold}
\end{figure}

In Fig. \ref{fig:Doppler_gain_MSE}, we evaluate the performance of the designed Doppler and channel gain estimation {scheme}. To focus solely on the accuracy of Doppler and channel gain estimation, we assume perfect acquisition of TO and delay information. We assess the mean square error (MSE) performance of the designed Doppler and channel gain estimation {scheme} and discuss the trend of MSE versus the value of ${\rm{SN}}{{\rm{R}}_{\rm{M}}}$. 
It can be observed that, similar to delay estimation, as ${\rm{SN}}{{\rm{R}}_{\rm{M}}}$ increases, the Doppler and channel gain estimation errors decrease.
\begin{figure}[h]
    \centering
    \includegraphics[width=0.4\textwidth]{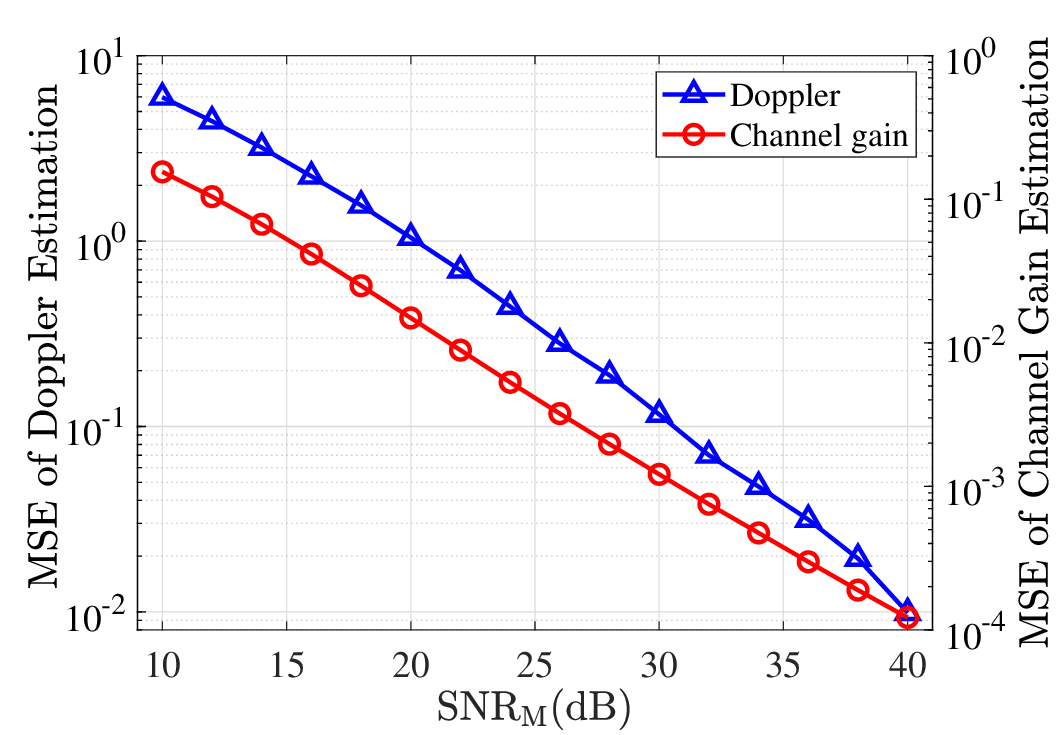}
    \caption{MSE of Doppler and channel gain estimation for the proposed algorithm.}
    \label{fig:Doppler_gain_MSE}
\end{figure}
\begin{figure}[h]
    \centering
    \includegraphics[width=0.4\textwidth]{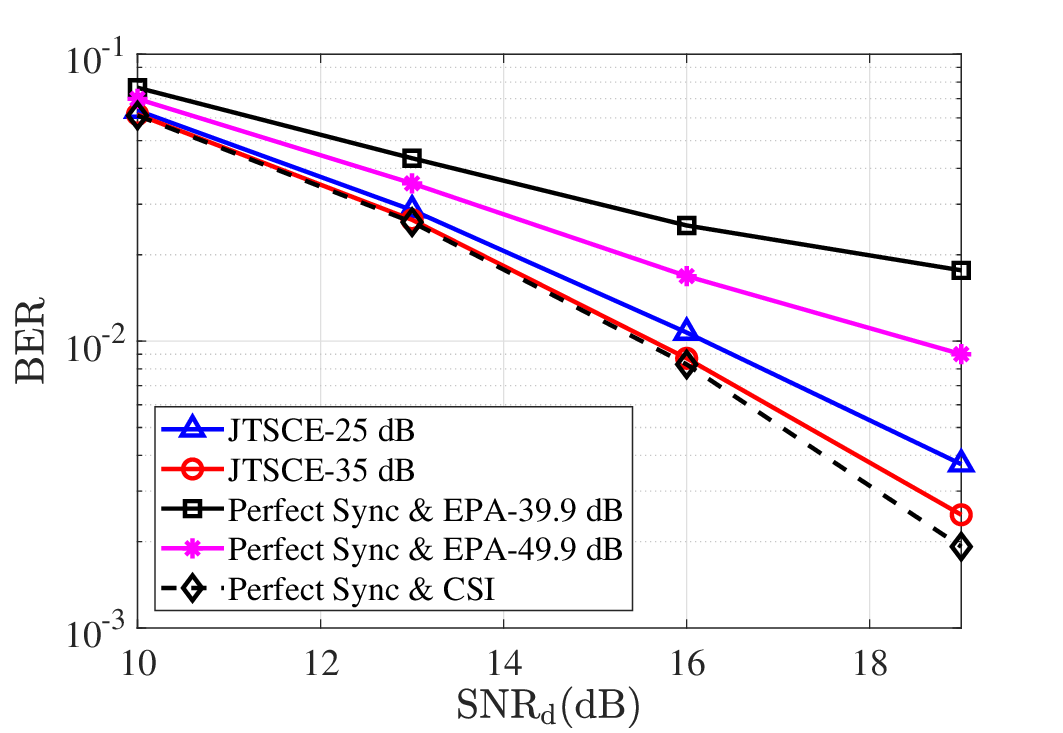}
    \caption{Comparison of BER performance of JTSCE, EPA with perfect time synchronization and perfect time synchronization \& CSI.}
    \label{fig:ber}
\end{figure}

Finally, in Fig. \ref{fig:ber}, we compare the proposed JTSCE with EPA to evaluate its BER performance. For signal detection, we employ the linear minimum mean-square error (LMMSE) algorithm {\cite{hong2022delay}}. 
To ensure accurate estimation of TO and delays, we set ${\rm{SN}}{{\rm{R}}_{\rm{M}}}$ to 25 dB and 35 dB based on the conclusions drawn from Fig. \ref{fig:to_delay_acc}. For a fair comparison, we ensure that the total power of the embedded sequences in both algorithms is the same. That is, set the pilot power {of EPA to be ${\log _{10}}\left( {N - 1} \right){\rm{\ dB}}$ greater than that of JTSCE.} 
The BER curve with perfect synchronization and CSI is also plotted. 
Simulation results show that as the power of pilot increases, the BER performance of both algorithms improves, which aligns with the improved accuracy of Doppler and gain estimation observed in Fig. \ref{fig:Doppler_gain_MSE}. 
Additionally, when allocating the same pilot power, JTSCE achieves better BER performance than EPA. 
{When ${\rm{SN}}{{\rm{R}}_{\rm{M}}}$ reaches 35 dB, the {achieved BER} of JTSCE is close to that of the case of perfect time synchronization and CSI.}

\section{Conclusion}
In this letter, we proposed a joint time synchronization and channel estimation algorithm for OTFS. In this algorithm, we leveraged MLS as pilots to accomplish time synchronization and channel estimation. Initially, we utilized the good autocorrelation properties of the MLS to estimate TO and delay.
This approach had the advantage of not requiring additional time-dimension TO estimation and maintaining TO estimation accuracy when the first path in the channel was not the strongest.
After determining the starting position of OTFS signal and the delay of each path of the channel, we estimated the Doppler and channel gain for each path. We analyzed the accuracy of the proposed algorithm's TO and delay estimation {performance} through simulations, as well as the MSE of Doppler and channel gain estimation. We also evaluated the BER performance, and the results indicated that the proposed algorithm's BER performance was close to the perfect time synchronization and CSI scenario. Moreover, with the same pilot power, it outperformed the embedded pilot-aided counterpart.

\appendix[Proof of \textbf{Lemma} \ref{lemma:1}]
When ${\theta _{\rm{t}}} \ne 0$ and ${\theta _{\rm{d}}} = 0$, ${Y_{\tilde n}}\left[ n \right]$ can be expressed as $Y_{\tilde{n}}\left[n\right]=\begin{cases}y_{\mathrm{TO}}\left[n\right],&n\leq\theta_{\mathrm{t}},\\y_{\mathrm{MLS}}^{(l_i)}\left[n-\theta_{\mathrm{t}}\right],&n>\theta_{\mathrm{t}}\end{cases}$, where ${y_{{\rm{TO}}}}\left[ n \right]$ represents the interference caused by ${\theta _{\rm{t}}}$. ${y_{{\rm{TO}}}}\left[ n \right]$ does not correlate with ${Y_{\tilde n}}\left[ n \right]$. Thus, we consider the case of ${\theta _{\rm{t}}} \ll N$.

The MLS ${x_{{\rm{MLS}}}}\left[ n \right]$ has the property $\sum\nolimits_{n = 0}^{N - 2} {{x_{{\rm{MLS}}}}\left[ n \right]{x_{{\rm{MLS}}}}\left[ {{{\left( {n - {\theta _{\rm{t}}}} \right)}_{N - 1}}} \right]}  =  - \frac{{{P_{{\rm{MLS}}}}}}{{N - 1}}$, for ${\theta _{\rm{t}}} \ne 0$. 
Therefore, we have 
\begin{equation}\label{eq:MLS_property}
\begin{array}{l}
\left| {\sum\nolimits_{n = {\theta _{\rm{t}}}}^{N - 2} {{x_{{\rm{MLS}}}}\left[ n \right]{x_{{\rm{MLS}}}}\left[ {n - {\theta _{\rm{t}}}} \right]} } \right|\\
 = \left| { - \frac{{{P_{{\rm{MLS}}}}}}{{N - 1}} - \sum\nolimits_{n = 0}^{{\theta _{\rm{t}}} - 1} {{x_{{\rm{MLS}}}}\left[ n \right]{x_{{\rm{MLS}}}}\left[ {{{\left( {n - {\theta _{\rm{t}}}} \right)}_{N - 1}}} \right]} } \right|\\
 \le \frac{{\left( {1 + {\theta _{\rm{t}}}} \right){P_{{\rm{MLS}}}}}}{{N - 1}}.
\end{array}
\end{equation}

Performing DFT on
\begin{equation}
\begin{aligned}q_{\tilde{n}}\left[n\right]&=y_{\mathrm{MLS}}^{(l_{i})}\left[n\right]\cdot\tilde{x}_{\mathrm{MLS}}\left[n\right]\\&=\begin{cases}y_{\mathrm{TO}}\left[n\right]\cdot\tilde{x}_{\mathrm{MLS}}\left[n\right]&n\leq\theta_{\mathrm{t}}\\\\h_{i}\cdot\tilde{x}_{\mathrm{MLS}}\left[n\right]\tilde{x}_{\mathrm{MLS}}\left[n-\theta_{\mathrm{t}}\right]\cdot\mathrm{e}^{\mathrm{j}\frac{2\pi k_{i}\left(n-\theta_{\mathrm{t}}\right)}{N}}&n>\theta_{\mathrm{t}}\end{cases}\end{aligned},
\end{equation}
we get 
\begin{equation}
\begin{array}{l}
{Q_{\tilde n}}\left[ k \right] = \frac{1}{{\sqrt N }}\sum\nolimits_{n = 0}^{{\theta _{\rm{t}}} - 1} {{y_{{\rm{TO}}}}\left[ n \right] \cdot {{\tilde x}_{{\rm{MLS}}}}\left[ n \right]{{\rm{e}}^{{\rm{ - j}}\frac{{2{\rm{\pi }}nk}}{N}}}} \\
 + \frac{1}{{\sqrt N }}\sum\nolimits_{n = {\theta _{\rm{t}}}}^{N - 2} {{h_i} \cdot {x_{{\rm{MLS}}}}\left[ n \right]{x_{{\rm{MLS}}}}\left[ {n - {\theta _{\rm{t}}}} \right] \cdot {{\rm{e}}^{{\rm{j}}\frac{{2{\rm{\pi }}{k_i}\left( {n - {\theta _{\rm{t}}}} \right)}}{N}}}{{\rm{e}}^{{\rm{ - j}}\frac{{2{\rm{\pi }}nk}}{N}}}} \\
 \approx \frac{1}{{\sqrt N }}\sum\nolimits_{n = {\theta _{\rm{t}}}}^{N - 2} {{h_i} \cdot {x_{{\rm{MLS}}}}\left[ n \right]{x_{{\rm{MLS}}}}\left[ {n - {\theta _{\rm{t}}}} \right] \cdot {{\rm{e}}^{{\rm{j}}\frac{{2{\rm{\pi }}{k_i}\left( {n - {\theta _{\rm{t}}}} \right)}}{N}}}{{\rm{e}}^{{\rm{ - j}}\frac{{2{\rm{\pi }}nk}}{N}}}} 
\end{array}.
\end{equation}

When $k = {k_i}$, according to (\ref{eq:MLS_property}), we obtain 
\begin{equation}
\begin{array}{c}
\left| {{Q_{\tilde n}}\left[ {{k_i}} \right]} \right| = \frac{{\left| {{h_i}} \right|}}{{\sqrt N }}\left| {\sum\nolimits_{n = {\theta _{\rm{t}}}}^{N - 2} {{x_{{\rm{MLS}}}}\left[ n \right]{x_{{\rm{MLS}}}}\left[ {n - {\theta _{\rm{t}}}} \right]} } \right|\\
 \le \frac{{\left| {{h_i}} \right|}}{{\sqrt N }}\frac{{\left( {1 + {\theta _{\rm{t}}}} \right){P_{{\rm{MLS}}}}}}{{N - 1}} \ll \frac{{\left| {{h_i}} \right|}}{{\sqrt N }}{P_{{\rm{MLS}}}}
\end{array}.
\end{equation}

Thus, when ${\theta _{\rm{t}}} \ne 0$ and ${\theta _{\rm{d}}} = 0$, $\left| {{Q_{\tilde n}}\left[ k \right]} \right|$ does not exhibit a significant peak.

This completes the proof. 
\ifCLASSOPTIONcaptionsoff
  \newpage
\fi



\bibliographystyle{IEEEtran}
%
\bibliography{UAV_Channel_Gain_WCL}

\begin{thebibliography}{10}
\providecommand{\url}[1]{#1}
\csname url@samestyle\endcsname
\providecommand{\newblock}{\relax}
\providecommand{\bibinfo}[2]{#2}
\providecommand{\BIBentrySTDinterwordspacing}{\spaceskip=0pt\relax}
\providecommand{\BIBentryALTinterwordstretchfactor}{4}
\providecommand{\BIBentryALTinterwordspacing}{\spaceskip=\fontdimen2\font plus
\BIBentryALTinterwordstretchfactor\fontdimen3\font minus
  \fontdimen4\font\relax}
\providecommand{\BIBforeignlanguage}[2]{{%
\expandafter\ifx\csname l@#1\endcsname\relax
\typeout{** WARNING: IEEEtran.bst: No hyphenation pattern has been}%
\typeout{** loaded for the language `#1'. Using the pattern for}%
\typeout{** the default language instead.}%
\else
\language=\csname l@#1\endcsname
\fi
#2}}
\providecommand{\BIBdecl}{\relax}
\BIBdecl

\bibitem{DBLP:journals/wcl/AbidTK24}
M.~H. Abid, I.~A. Talin, and M.~I. Kadir, ``Wavelet-aided {OTFS} for
  {RIS}-assisted high-mobility wireless channels,'' \emph{{IEEE} Wireless
  Communications Letter}, vol.~13, no.~6, pp. 1611--1615, 2024.

\bibitem{Singh2024AutoencoderBE}
A.~Singh, S.~Sharma, M.~Sharma, K.~Deka, and D.~B. da~Costa, ``Autoencoder
  based end-to-end {OTFS} system design with hardware impairments,'' \emph{IEEE
  Wireless Communications Letters}, 2024, in Press. DOI
  10.1109/LWC.2024.3412240.

\bibitem{raviteja2019embedded}
P.~Raviteja, K.~T. Phan, and Y.~Hong, ``Embedded pilot-aided channel estimation
  for {OTFS} in delay--doppler channels,'' \emph{IEEE transactions on vehicular
  technology}, vol.~68, no.~5, pp. 4906--4917, 2019.

\bibitem{yuan2021data}
W.~Yuan, S.~Li, Z.~Wei, J.~Yuan, and D.~W.~K. Ng, ``Data-aided channel
  estimation for {OTFS} systems with a superimposed pilot and data transmission
  scheme,'' \emph{IEEE wireless communications letters}, vol.~10, no.~9, pp.
  1954--1958, 2021.

\bibitem{hashimoto2021channel}
N.~Hashimoto, N.~Osawa, K.~Yamazaki, and S.~Ibi, ``{Channel estimation and
  equalization for CP-OFDM-based OTFS in fractional doppler channels},'' in
  \emph{2021 IEEE International Conference on Communications Workshops (ICC
  Workshops)}.\hskip 1em plus 0.5em minus 0.4em\relax IEEE, 2021, pp. 1--7.

\bibitem{wei2022off}
Z.~Wei, W.~Yuan, S.~Li, J.~Yuan, and D.~W.~K. Ng, ``Off-grid channel estimation
  with sparse bayesian learning for {OTFS} systems,'' \emph{IEEE Transactions
  on Wireless Communications}, vol.~21, no.~9, pp. 7407--7426, 2022.

\bibitem{bayat2022time}
M.~Bayat and A.~Farhang, ``Time and frequency synchronization for {OTFS},''
  \emph{IEEE Wireless Communications Letters}, vol.~11, no.~12, pp. 2670--2674,
  2022.

\bibitem{li2023downlink}
S.~Li, M.~Zhang, C.~Ju, D.~Wang, W.~Chen, P.~Zhou, and D.~Wang, ``Downlink
  carrier frequency offset estimation for {OTFS}-based {LEO} satellite
  communication system,'' \emph{IEEE Communications Letters}, 2023.

\bibitem{hong2022delay}
Y.~Hong, T.~Thaj, and E.~Viterbo, \emph{Delay-Doppler Communications:
  Principles and Applications}.\hskip 1em plus 0.5em minus 0.4em\relax Academic
  Press, 2022.

\bibitem{yuan2020simple}
W.~Yuan, Z.~Wei, J.~Yuan, and D.~W.~K. Ng, ``A simple variational {Bayes}
  detector for orthogonal time frequency space ({OTFS}) modulation,''
  \emph{IEEE transactions on vehicular technology}, vol.~69, no.~7, pp.
  7976--7980, 2020.

\end{thebibliography}

\end{document}